\documentclass[11pt]{article}

\usepackage[normalem]{ulem}
\usepackage{amsmath}
\newcommand{\stkout}[1]{\ifmmode\text{\sout{\ensuremath{#1}}}\else\sout{#1}\fi}

\usepackage[utf8]{inputenc}
\usepackage{amsthm}
\usepackage{amsmath}
\usepackage{amssymb}
\usepackage{amsfonts}
\usepackage{mathtools}
\usepackage{algorithm}%http://ctan.org/pkg/algorithms
\usepackage{braket}
\usepackage{algpseudocode}
\usepackage{graphicx}
\usepackage{xcolor}
\usepackage{xspace}
\usepackage[margin=1in]{geometry}

\theoremstyle{plain}

\DeclarePairedDelimiter\ceil{\lceil}{\rceil}
\DeclarePairedDelimiter\floor{\lfloor}{\rfloor}

\newtheorem{lemma}{Lemma}
\newtheorem{theorem}{Theorem}
\newtheorem{problem}{Problem}

\theoremstyle{plain}

\definecolor{darkgreen}{rgb}{0,0.6,0}
\newcommand{\kibitz}[2]{\ifnum\Comments=1{\color{#1}{#2}}\fi}

\newcount\Comments  % 0 suppresses notes to selves in text
\title{The K-Centre Problem for Necklaces}
\author{D. Adamson, A. Deligkas, V. Gusev, I. Potapov}

\author{Duncan Adamson, Argyrios Deligkas, Vladimir V. Gusev, Igor Potapov}

\date{\today}

\begin{document}

% Format:
% 1 inch Margin, 11pt font, single column single space.
% Title page
% - Names
% - Emails
% - Affiliation
% - 1/2 paragraph abstract
% Main text
% - No page limit
% - Anything beyond the first 10 pages may not necessarily be considered

\maketitle

\begin{abstract}
In graph theory, the objective of the $k$-centre problem is to find a set of $k$ vertices for which the largest distance of any vertex to its closest vertex in the $k$-set is minimised. 
In this paper, we introduce the $k$-centre problem for sets of necklaces, i.e. the equivalence classes of words under the cyclic shift.
This can be seen as   
the $k$-centre problem on the complete weighted graph where every necklace is represented by a vertex, and each edge has a weight given by the overlap distance between any pair of necklaces.
Similar to the graph case, the goal is to choose $k$ necklaces such that the distance from any word in the language and its nearest centre is minimised. 
However, in a case of \emph{$k$-centre problem for languages} the size of associated graph maybe exponential in relation to the description of the language, i.e., the length of the words  $\ell$ and the size of the alphabet $q$.
We derive several approximation algorithms for the $k$-centre problem on necklaces,
with logarithmic approximation 
factor in the context of $\ell$ and $k$, and within a constant factor for a more restricted case.
%The problem is motivated by 
%the Extended Module Materials Assembly (EMMA) method used for in silico predictions of novel materials.
\end{abstract}

% \newpage

% \emph{State of the sections:}

% \begin{tabular}{l|l}
%     Abstract & First draft done.\\
%     1 & Unwritten. \\
%     2 & Ready to be checked. \\
%     2.1 & Ready to be checked. \\
%     2.2 & Ready to be checked.\\
%     2.3 & Ready to be checked.\\
%     3.1 & Ready to be checked.\\
%     3.2 & Ready to be checked.\\
%     3.3 & Ready to be checked.\\
%     4.1 & Ready to be checked.\\
%     4.2 & Ready to be checked.\\
%     4.3 & Ready to be checked.\\
%     5 & Unwritten.
% \end{tabular}

% \emph{What to consider for the appendix.}

% \begin{tabular}{l|l|l}
%     Section & Item & Discussion \\
%     \hline
%     2 & Figure 1 & Pictures make good candidates for moving to the appendix.\\
%     2 & Figure 2 & Pictures make good candidates for moving to the appendix.\\
%     3.1 & Proof of theorem 1 & Fairly straight forward proof following a description of the algorithm.\\
%     3.1 & Algorithm & A clear description may work better in the text.\\
%     3.3 & Proof the theorem. & Hopefully should be straightforward from the description.\\
%     4.1 & Section 4.1.1 & Follows fairly directly from the original case, without much novelty.
% \end{tabular}

% \emph{What letters stand for.}

% \begin{tabular}{l l}
%     $q$ & Size of the alphabet. \\
%     $k$ & Number of samples.\\
%     $\ell$ & Maximum length of the samples.\\
%     $j$ & (Normally) length of the longest subword (sometimes itterator).\\
%     $N_i$ & Length of a multidimensional word 
% \end{tabular}

% \newpage

\section{Introduction}
\label{sec:introduction}
In graph theory, the objective in $k$-centre problem is to find a set of $k$ vertices for which the largest distance of any vertex of the graph and its closest vertex in this $k$-set is minimised. 
The numerous applications of the problem in various areas of computer science, lead to different definitions of connectivity and distance between the vertices depending on the application at hand.
%It has application in various areas of computer science and depending on the application a connectivity between vertices or a distance between vertices can be defined in a different way. 
The $k$-centre problem on graphs is known to be NP-hard. 
The best performance ratio for a polynomial-time approximation solution is $2$ unless P = NP, but it is  unlikely to be fixed-parameter tractable (FPT) in a context of the most natural parameter $k$, which is the number of centres \cite{Algorithmica2020}.

A different form of the $k$-centre problem
appears in stringology and it was linked with important applications in computational biology for example to find the approximate gene clusters for a set of words over the DNA alphabet \cite{JACM2002}. This problem is also NP-hard problem; there are fixed-parameter algorithms and heuristic algorithms for it without any performance guarantee. The closest string problem aims to find a new string within a distance $d$ to each input of $n$ strings and such that $d$ is minimized.
The natural generalization of $k$-Closest String problem is of finding $k$-centre strings of a given length minimizing the distance from every string to closest
centre \cite{SODA99,CPM2004}. 
This problem has been mainly studied for the most popular distance measure which is the Hamming distance. 
The major application of this distance is in the coding theory, but it also has been intensively used in many biological applications which aim to discover 
a region of similarity or to
design probes or primers \cite{IC2003}. 

In this paper we define and study a new variant of the $k$-centre problem 
on the new objects, the class of {\sl necklaces}, and using a different distance function, {\sl the overlap coefficient} to define the closeness of strings or necklaces. 
{\sl Necklaces}  are  classical structures in combinatorics, which can be defined as a set of $\ell$-character strings over an alphabet of size $q$, that are equivalent under the cyclic shift operation.
The research on necklaces in combinatorics has been mainly focused on characterisation of these objects, their efficient generation and comparison, ordering as well as on the design of an  efficient ranking and unranking procedures \cite{Kopparty2016,Sawada2017}.
In this paper, we 
%move forward by 
link several problems and interconnect research ideas on these fascinating objects as well as design approximation algorithms for the $k$-centres problem on necklaces, which have applications for combinatorial crystal structure prediction
\cite{SOFSEM2020,CSP2020}.

In particular we are motivated to study {\sl $k$-centre problem} for a class of necklaces 
by the {\sl Extended Module Materials Assembly (EMMA)} method used for in silico predictions of novel materials~\cite{ Collins17,Dyer13}. 
The idea of this method is to consider materials assembled from well-chosen layers, where the arrangements with low enough energy constitute potentially stable materials. Since the space of all potential stackings is typically too large for exhaustive search, one is looking for a diverse and representative sample of this space to be used with further optimisation strategies. To approach this problem, we model layers as letters and materials (periodic crystals) as necklaces due to their invariance under the cyclic shift operation. The set of necklaces during the optimisation is often further constrained: fixed chemical composition corresponds to necklaces with the fixed Parikh vector and constraints on the relative position of layers lead to necklaces with forbidden subwords.
These considerations lead us to the general problem of sampling from languages: 
$\mathcal{L}^{\Sigma}_{\ell}$  -- the set of all words of length $\ell$;
$\mathcal{L}^{\leq\Sigma}_{\ell}$ -- the set of all words of length at most $\ell$; $\mathcal{L}^{\Sigma}_{\mathbf{P}}$ -- the set of all words with the Parikh vector $\mathbf{P}$  and $\mathcal{L}^{\Sigma\setminus F}_{\ell}$ --
the set of all words of length at most $\ell$ that do not contain words from a finite set $\mathcal{F}$ as factors. 
%   denoted by $\Sigma^{\leq \ell}\setminus \mathcal{F}$. 
%We consider three cases for the language $\mathcal{L}$:
%
%\begin{enumerate}
%    \item the set of all words of length at most $\ell$ denoted by 
%$\Sigma^{\leq \ell}$;
%    \item the set of all words with the Parikh vector $\vec{p}$ denoted by 
%$\Sigma^\mathbf{p}$;
%    \item the set of all words of length at most $\ell$ that do not contain 
%words from a finite set $\mathcal{F}$ as factors 
%    denoted by $\Sigma^{\leq \ell}\setminus \mathcal{F}$.
%\end{enumerate}
Apart from the Hamming distance, there are several well known methods
for comparing words with their own advantages and disadvantages \cite{cohen2003comparison, piskorski2007comparison, recchia2013comparison}.
In order to define the closeness between different necklaces we use
one of such methods based on computing the overlap coefficient between each pair of necklaces. In particular, in the context of the material science,  two patterns of layers may have closer properties if they have more common fragments.

In general the $k$-centre problem on necklaces can be seen as 
the $k$-centre problem on the complete weighted graph where every word is represented by a node, and each edge has a weight given by the overlap distance between the two words.
As in the graph case, the goal  is to choose $k$ words such that the distance from any word in the language and its nearest centre is minimised. 
However, in a case of  \emph{$k$-centre problem for languages} the
associated graph can be of exponential size in the context of the input.
%
%However in a case of  \emph{$k$-centre problem for languages} the size of associated graph maybe exponential in relation to the description of the language, i.e., the input size.

The main results of this paper is in the design of approximation algorithms for several finite languages of necklaces $\mathcal{L}^{\Sigma}_{\ell}$, $\mathcal{L}^{\Sigma}_{\mathbf{P}}$ and $\mathcal{L}^{\Sigma\setminus F}_{\ell}$
with logarithmic approximation 
factor in the context of $\ell$ and $k$ and log-linear for $\mathcal{L}^{\leq\Sigma}_{\ell}$. 
The first algorithm is based on building a prefix tree of necklaces utilising previously designed ranking and unranking procedures. 
Then we extend an approximation algorithm for a language of necklaces with forbidden words by designing new procedures for ranking and unranking of necklaces under forbidden words and Parikh map constraints, 
%\vl{no material science here} 
where both limitations are motivated by natural material science constraints. Finally we propose a different technique based on building $k$ centres for $\mathcal{L}^{\Sigma}_{\ell}$ 
via the de Bruijn sequences, which can find a solution in linear time with a constant approximation factor.

\section{Preliminaries}
\label{sec:prelims}

%A \emph{language} $\mathcal{L}$ is a set of words defined by rules applied to an alphabet.

A formal \emph{language} $\mathcal{L}$ consists of words whose letters are taken from an alphabet and are defined according to a specific set of rules.
We focus on \emph{cyclic languages} languages consisting of cyclic words {\em only}.
A cyclic word is the equivalence class of words under the \emph{cyclic shift} operation, also known as {\em necklace}.
A cyclic shift of size $i$ moves the suffix of length $i$ to the front of the word, while maintaining the relative order within the suffix.
More formally, the cyclic shift of length $i$ on the word $w = w_1 w_2 \hdots w_n$ will transform it to $w_{n - i + 1} \hdots w_n w_1 \hdots w_{n - i}$.
Any word that is a member of this equivalence class is a \emph{representation} of the cyclic word.
In general, cyclic words are represented by the lexicographically smallest word in this equivalence class, known as the \emph{canonical representation}.
{\em Lyndon words} 
%form an important subclass of necklaces and it is 
form the set of aperiodic necklaces;  necklaces such that given the canonical form $w_1 w_2 \hdots w_\ell$, there exists no cyclic shift such that $w_i w_{i + 1} \hdots w_{\ell} w_1 \hdots w_{i - 1} = w_1 \hdots w_\ell$ for any $i \neq 1$.
%This class of necklaces is known as the set of Lyndon words.

The class of \emph{fixed length cyclic languages} consists of all cyclic words from an alphabet $\Sigma$, made of any combination of characters of $\Sigma$ with length $\ell$.
This is equivalent to the set of necklaces of length $\ell$ over the alphabet $\Sigma$.
This language will be denoted $\mathcal{L}_{\ell}^{\Sigma}$.
% A generalisation of $\mathcal{L}_{\ell}^{\Sigma}$ is the \emph{maximum length cyclic language}, the language made of all cyclic words of length at most $\ell$.
% This will be denoted $\mathcal{L}_{\leq \ell}^{\Sigma}$. Note that $\mathcal{L}_{\leq \ell}^{\Sigma} = \cup_{i = 0}^{\ell} \mathcal{L}_{i}^{\Sigma}$.
% When the alphabet $\Sigma$ is clear from the context, we will write  $\mathcal{L}_{\ell}$ and $\mathcal{L}_{\leq \ell}$ instead of
% $\mathcal{L}
We focus on two restricted cases of this language.
% As well as this general language, two more restricted cases will be considered.

The first is the \emph{fixed length cyclic language with forbidden subwords}. This is a fixed length cyclic language
where any word does not contain subwords 
from a set of forbidden words.
% without any word containing some forbidden subword from a given set.
Given the set of forbidden words $F$, the language of all words of length $\ell$ without these will be denoted $\mathcal{L}_{\ell}^{\Sigma \setminus F}$.
Formally a word $w \in \mathcal{L}_{\ell}^{\Sigma}$ will be in $\mathcal{L}_{\ell}^{\Sigma \setminus F}$ if there is no representation of it of the form $w_1 \hdots w_i f_1 \hdots f_j w_{i + j + 1} \hdots w_{\ell}$ for any word $f_1 \hdots f_j \in F$.
% This may also be extended to the \emph{maximum length cyclic language without forbidden subwords}, denoted $\mathcal{L}_{\leq \ell}^{\Sigma \setminus F}$.
% \argy{why do we define it for both {\em specific} languages? We can define the problem for any language $\mathcal{L}$ and then say that we will focus on the two specific categories. This way we can save some space.}

The second language is the \emph{fixed content cyclic language}.
Here, in any word the number of occurrences of each letter of $\Sigma$ is fixed.
The number of occurrences of each character will be given as a vector $\mathbf{P}$, where $\mathbf{P}_i$ denotes the number of occurrences of the $i^{th}$ character.
This language will be denoted $\mathcal{L}_{\mathbf{P}}^{\Sigma}$.
% Since the number of appearances of each character is fixed, the length of the words in this language is implicitly fixed by the sum of the occurrences of each character.

A generalisation of $\mathcal{L}_{\ell}^{\Sigma}$ and $\mathcal{L}_{\ell}^{\Sigma \setminus F}$ is to \emph{maximum length languages}.
These contain every word of length less than or equal to $\ell$ in the corresponding fixed length language.
For a given $\ell$, the maximum length generalisation of $\mathcal{L}_{\ell}^{\Sigma}$ will be denoted $\mathcal{L}_{\leq\ell}^{\Sigma}$ and $\mathcal{L}_{\ell}^{\Sigma \setminus F}$ will be denoted $\mathcal{L}_{\leq\ell}^{\Sigma \setminus F}$.
Formally, this can be written as $\mathcal{L}_{\leq \ell}^{\Sigma} = \cup_{i = 1}^{\ell} \mathcal{L}_{i}^{\Sigma}$ and $\mathcal{L}_{\leq \ell}^{\Sigma \setminus F} = \cup_{i = 1}^{\ell} \mathcal{L}_{i}^{\Sigma \setminus F}$.
As a cyclic word of length $\ell$ can be seen as a word of infinite length with a period of a most $\ell$, when comparing two cyclic words it makes sense to look at two representatives of these words with the same length, which 
%These 
will be called the \emph{same length representatives} of the words.

% $\mathcal{L}_{\leq \ell}^{\Sigma}$, will correspond to the union of each language $\mathcal{L}_l^{\Sigma}$ for every $l \leq \ell$.

%\paragraph*{The Overlap Coefficient}
\noindent
{\bf The Overlap Coefficient.} 
The \emph{Overlap coefficient} of the sets $A$ and $B$ is defined as the size of the intersection of the two sets, normalised by the size of the smaller set, i.e. $\mathfrak{C}(A,B) = \frac{|A \cap B|}{\min(|A |,|B|)}$.
For the overlap coefficient measures the closeness of two sets $A$ and $B$, where a value of $1$ means that the two sets are identical, and a value of $0$ means there are no shared elements.

The Overlap coefficient $\mathfrak{C}(\alpha,\beta)$ for two cyclic words $\alpha$ and $\beta$ is defined as the overlap coefficient between the multisets of all subwords of the same length representatives of $\alpha$ and $\beta$.
Given the same length representative $a$ some word $\alpha$ of length $p \cdot \ell$, the multiset of subwords of length $l$ is the multiset of each subword starting at every position from $1$ to $p \cdot \ell$, labelled by the number of occurrences of this subword up to this point.
Note that as this word is cyclic, subwords of length $l$ may occur beginning in the last $l - 1$ positions of the word, giving a total of $p \cdot \ell$ subwords of length $l$ for any $l$.
For example, given the word $aaab$, the multiset of subwords of length 2 are $\{aa_1,aa_2,ab_1,ba_1\}$.
The multiset of all subwords is simply  the union of the multisets of the subwords for every length from 1 to $p \cdot \ell$, having a total size of $(\ell \cdot p)^2$.
An example of this is given explicitly between $ab$ and $abb$ in Figure \ref{fig:overlap_example}.

To use this as a distance, the measure will be inverted so that a value of $1$ will imply the strings share no similarity and a value of $0$ implies they represent the same word.
From this, the Overlap distance between two cyclic words $\alpha$ and $\beta$ will be given by
\[
\mathfrak{O}(\alpha, \beta) = \begin{cases} 
    \infty, & \text{if~}\mathfrak{C}(\alpha, \beta) = 0\\
    0 & \text{if~} \mathfrak{C}(\alpha, \beta) = 1\\
    \frac{1}{\mathfrak{C}(\alpha, \beta)} & \text{Otherwise.}
\end{cases}
\]
\noindent
{\bf The k-Centre Problem on Necklaces.}
With this distance, the \emph{$k$-centre problem for languages} can be defined.
This can be thought of as the $k$-centre problem on the complete weighted graph where every word is represented by a node, and each edge has a weight given by the overlap distance between the two words.
As in the graph case, the goal here is to choose $k$ words such that the distance from any word in the language and its nearest centre is minimised. However in a case of  \emph{$k$-centre problem for languages} the size of associated graph maybe exponential in relation to the description of the language, i.e., the input size.
%A formal definition is given in Problem \ref{prob:k_sample}.

\begin{figure}
    \centering
 %   \begin{tabular}{l l | l l}
 %       word & same length representative 
 %& word & same length representative\\
 %       $ab$ & $ababab$ & $abb$ & 
 %$abbabb$\\
 %   \end{tabular}
    \scriptsize{
    \begin{tabular}{l|l|l}
        %Length 
        & word $ab$ with representative $(ab)^3$ 
        %, the \emph{same length representative}  $ababab$  
        & word $abb$ with representative $(abb)^2$ \\ 
        %, the \emph{same length representative} $abbabb$\\
        \hline
        1 & $\mathbf{a}_1, \mathbf{b}_1, \mathbf{a}_2, \mathbf{b}_2, a_3, \mathbf{b}_3$ & $\mathbf{a}_1, \mathbf{b}_1, \mathbf{b}_2, \mathbf{a}_2, \mathbf{b}_3, b_4$\\
        2 & $\mathbf{ab}_1, \mathbf{ba}_1,\mathbf{ab}_2, \mathbf{ba}_2,ab_3, ba_3$ & $\mathbf{ab}_1, bb_1,\mathbf{ba}_1, \mathbf{ab}_2, bb_2, \mathbf{ba}_2$\\
        3 & $aba_1, \mathbf{bab}_1,aba_2, \mathbf{bab}_2,aba_3, bab_3$ & $abb_1, bba_1, \mathbf{bab}_1, abb_2, bba_2, \mathbf{bab}_2$\\
        4 & $abab_1, baba_1,abab_2, baba_2,abab_3, baba_3 $ & $abba_1, bbab_1, babb_1, abba_2, bbab_2, babb_2$\\
        5 & $ababa_1,babab_1,ababa_2,babab_2,ababa_3, $ & $abbab_1, bbabb_1, babba_1,abbab_2, bbabb_2,$\\
          & $babab_3$ & $babba_2$\\
        6 & $ababab_1, bababa_1,ababab_2, bababa_2,$ & $abbabb_1, bbabba_1, babbab_1,abbabb_2$\\
          & $ababab_3, bababa_3$ & $bbabba_2, babbab_2$
    \end{tabular}
    }
    
 %   \begin{align*}
 %       \mathfrak{O}(ab,abb) = 1 - \frac{11}{36} = \frac{25}{36}
 %   \end{align*}
    \caption{Example of the overlap coefficient calculation for a pair of words $ab$ and $abb$. 
    There are $11$ common subwords out of the total number of $36$ subwords of length from $1$ till $6$ in the same length cyclic words representatives $(ab)^3$ and $(abb)^2$, so $\mathfrak{C}(ab,abb)= \frac{11}{36}$
    and  $\mathfrak{O}(ab,abb) =  \frac{36}{11}$.
    %= \frac{25}{36}
    }
    \label{fig:overlap_example}
\end{figure}

\begin{figure}
    \centering
    \begin{tabular}{l l l l l l}
        A & $aaaa$ & B & $aaab$ & C & $aabb$  \\
        D & $abab$ & E & $abbb$ & F & $bbbb$
    \end{tabular}
 %   Overlap between:
    \begin{tabular}{l| l l l l l l}
         $\alpha \backslash \beta$  & A & B & C & D & E & F \\
        \hline
        A & 0 & $\frac{16}{6}$ & $\frac{16}{3}$ & $8$ & $16$ & $\infty$\\
        B & $\frac{16}{6}$ & 0 & $\frac{16}{7}$ & $\frac{16}{6}$ & 4 & $16$\\
        C & $\frac{16}{3}$ & $\frac{16}{7}$ & 0 & $\frac{16}{6}$ & $2$ & $\frac{16}{3}$\\
        D & $8$ & $\frac{16}{6}$ & $\frac{16}{6}$ & 0 & $\frac{16}{10}$ & $8$\\
        E & $16$ & $4$ & $2$ & $\frac{16}{6}$ & 0 & $\frac{16}{6}$\\
        F & $\infty$ & $16$ & $\frac{16}{3}$ & $8$ & $2$ & 0
    \end{tabular}
    
    \caption{Example of the overlap distance $\mathfrak{D}(\alpha,\beta)$ for binary cyclic words of length 4.}
    \label{fig:overlapDistances}
\end{figure}

% In this problem, 
% Here we will consider the \emph{full cyclic intersection index}, where every subword of the same-length representatives of the words will be considered, normalised by the size of the sets.
% For example, for words with same-length representatives of length $n$, there are $n^n$ subwords.
% We denote the full cyclic intersection index between two words $\alpha$ and $\beta$ as $fc(\alpha, \beta)$ $\mathfrak{F}(\alpha,\beta)$.
% We will considered the \emph{$r$-bounded cyclic intersection index}, where every subword of length at most $r$ of the \emph{same length representatives} of the words will be considered, normallised by the size of the sets.
% Note that for words with same length representatives of length $n$, there will be $n^r$ subwords.
% We denote the bounded cyclic intersection index as $bc(r,\alpha, \beta)$.
%
%\begin{tabular}{|p{5in}|}
%    \hline
 % 
    \begin{problem}
        {\bf $k$-Centre problem for languages:}
     %{\bf Input:} 
     Given a finite language $\mathcal{L}$ and an integer $k$,
%     
%     {\bf Task:} 
select $k$ words from $\mathcal{L}$ forming a sample $S$ of $k$-centres, minimising the maximum overlap distance between every word in $\mathcal{L}$ and the nearest member of $S$:
     \begin{align}
        \mathfrak{D}_{\mathcal{L}, k} = \min_{|S| = k} \max_{v \in \mathcal{L}} \min_{s \in S} \mathfrak{O}(s,v).
    \end{align}
     \label{prob:k_sample}
    \end{problem} 
%    \\
%    \hline
%\end{tabular}
Our goal is to  maximise the length $\lambda$ of the longest subword such that every word in $\mathcal{L}$ shares a subword of length $\lambda$ with at least one member of the sample.
The idea behind this approach is that if a word shares a subword of length $\lambda$ with the sample, it will also share 2 words of length $\lambda - 1$, 3 of length $\lambda - 2$, and so on, for a total of $\frac{\lambda(\lambda + 1)}{2}$ common subwords.
Therefore by increase the length of these subwords by $1$, there is a quadratic increase in the size of the intersection in the overlap coefficient.
This provides a bound on the maximum distance between every word in $\mathcal{L}$ and the centres is created of $\mathfrak{D}_{\mathcal{L}, k} \leq \frac{\ell^2}{\lambda(\lambda + 1)}$.
So our algorithms will be focused 
on maximising the length of $\lambda$ for the languages $\mathcal{L}^{\Sigma}_{\ell}$, $\mathcal{L}_{\leq \ell}^{\Sigma}$, $\mathcal{L}^{\Sigma \setminus F}_{\ell}$, $\mathcal{L}_{\leq \ell}^{\Sigma \setminus F}$, and $\mathcal{L}_{\ell}^{\Sigma}$.
%
%The remainder of this work will focus on algorithms maximising the length of $\lambda$ for the languages $\mathcal{L}^{\Sigma}_{\ell}$, $\mathcal{L}_{\leq \ell}^{\Sigma}$, $\mathcal{L}^{\Sigma \setminus F}_{\ell}$, $\mathcal{L}_{\leq \ell}^{\Sigma \setminus F}$, and $\mathcal{L}_{\ell}^{\Sigma}$.

\begin{lemma}
For the language $\mathcal{L}_{\ell}$, given $\lambda$ as longest length such that every subword in $\mathcal{L}$ shares a subword of length $\lambda$ with at least one centre, $\mathfrak{D}_{\mathcal{L}_{\ell}, k} \leq \frac{2\ell^2}{\lambda(\lambda + 1)}$.
\label{lem:max_diistance}
\end{lemma}
\begin{proof}
Let us assume 
%By the definition of $\lambda$, 
that every word $w \in \mathcal{L}$ 
%must 
shares at least one subword of length $\lambda$
with at least one centre $w_c$.
Thus, $w_c$ will contain all subwords of this shared word, giving a total of $\frac{\lambda(\lambda + 1)}{2}$ shared subwords.
This gives an intersection of $w$ with the closest member of the sample set of size at least $\frac{\lambda(\lambda + 1)}{2}$.
As the size of the multiset of subwords (in cyclic words of length $l$) will be $\ell^2$, the total distance will be $\frac{2\ell^2}{\lambda(\lambda + 1)}$.
% In the case the length of the words is fixed, then this gives the overlap distance as being at least $1 - \frac{\lambda(\lambda - 1)}{2\ell^2}$.
% In the more general case, where $\ell$ is an upper bound on the distance, the same length representatives must be considered.
% Clearly the longest same length representative will be of size $\ell(\ell - 1)$.
% However, this will lead to the shared subword of length $\lambda$ being repeated $\ell - 1$ times in one word, and $\ell$ times in the other, increasing the size of the intersection to $\frac{(\ell - 1)\lambda(\lambda - 1)}{2}$.
% Cancelling out the $(\ell - 1)$ terms in the numerator and denominator gives an upper bound on the overlap distance of $1 - \frac{\lambda(\lambda - 1)}{2\ell^ 2(\ell - 1)}$.
\end{proof}

\begin{lemma}
For the language $\mathcal{L}^{\Sigma}_{\ell}$ and a sample $S$ of $k$-centres
%\stkout{of length at most $\ell$}, 
the optimal overlap distance $\mathfrak{D}_{\mathcal{L}, k}$
between every word in $\mathcal{L}$ and the nearest member of $S$ is no less than
%greater or equal to 
$\frac{\ell^2}{\log_q(\ell^2k)(\log_q(\ell^2k) + 1)}$.
\label{lem:min_diistance}
\end{lemma}
\begin{proof}
Let $\lambda$ be the longest subword such that every word in $\mathcal{L}$ shares a subword of length $\lambda$ with at least one of the $k$ centres.
To get an upper bound on the size of the $\lambda$, observe that every word in $\mathcal{L}$ is a necklace, therefore every word must have at least one subword of length $\lambda$ that is a prefix of a necklace.
Therefore the longest value for $\lambda$ is the largest value such that there are fewer than $k \times \ell$ necklace prefixes. 
A simple lower bound on the number of necklace prefixes of length $\lambda$ for an alphabet of size $q$ is $\frac{q^\lambda}{\lambda}$.
This can be rewritten as an inequality in terms of $k$ and $\ell$ as $\frac{q^\lambda}{\lambda} \leq k \ell$.
Observing that $\lambda \leq \ell$ gives $\frac{q^{\lambda}}{\ell} \leq k \ell$, giving as a bound on $\lambda$, $\lambda \leq \log_{q}(\ell^2 k)$.

Assume that the furthest word share two subwords of length $\lambda$ with the nearest centre.
For this to be the case, every centre must contain every subword of length $\lambda$, with at least one occurring twice, requiring the string to be of length $q^{\lambda} + 1$, which is clearly greater than $\ell$ under the assumption that $\lambda = \log_q(\ell^2k)$.
Using this as an upper bound, the intersection may be of size no more than $\lambda(\lambda + 1)$, giving a distance of $\frac{\ell^2}{\lambda(\lambda + 1)}$.
Using the upper bound on $\lambda$ gives a lower bound on the distance of $\frac{\ell^2}{\log_q(\ell^2k)(\log_q(\ell^2k) + 1)}$.
\end{proof}

\begin{lemma}
Given an algorithm that can approximate the solution to Problem \ref{prob:k_sample} for $\mathcal{L}_{\ell}^{\Sigma}$ within a factor $f$, the same sample will be an approximation of $(\ell - 1)f$ of the optimal solution to Problem 1 for $\mathcal{L}_{\leq\ell}^{\Sigma}$.
\label{lem:max_length_to_fixed_length}
\end{lemma}

\begin{proof}
Let $\lambda$ be the length of the longest subword such that every word in $\mathcal{L}_{\ell}^{\Sigma}$ shares a subword of length $\lambda$ with at least one centre in the sample.
Observe that for every length $l \leq \ell$, every word in $\mathcal{L}_l^{\Sigma}$ will occur as a subword of at least $q$ words in $\mathcal{L}_{\ell}^{\Sigma}$.
As the samples must include every subword of length $\lambda$ as a subword, any word of length $\lambda < l < \ell$ will also share a common subword of length $\lambda$ with each centre, which when converted to the same length representatives of length $l\cdot\ell$ will lead to an overlap coefficient of $\frac{l\lambda(\lambda + 1)}{2(l\ell)^{2}}$.
For words of length $l < \lambda$, observe that the same length representative with a word of length $\ell$ will also be of length at least $\ell$, ensuring that it shares a common subword of length at least $\lambda$.
As in the case $l \geq \lambda$, the overlap coefficient between some word with length $l < \lambda$ will be $\frac{l\lambda(\lambda + 1)}{2(l\ell)^{2}}$

Dividing this by the bound given in Lemma \ref{lem:max_diistance} gives $\frac{2l\lambda(\lambda + 1)\ell^2}{2\ell^2l^2\lambda(\lambda + 1)} = \frac{1}{l}$.
The worst case for this will be $l = \ell - 1$.
Therefore, the solution for an algorithm guaranteeing an approximation factor of $f$ for Problem \ref{prob:k_sample} on the language $\mathcal{L}_{\ell}^{\Sigma}$ will give a solution that is an approximation of the optimal solution by a factor of $(\ell - 1)f$.
\end{proof}

\section{Sampling via Prefix Trees}

In this section we will look at a generic framework for sampling necklaces under various constraints.
This will give a logarithmic approximation factor relative to the number of samples in the general case.
In Section \ref{sec:prefix_tree} we will present the algorithm and show how it preforms on the languages $\mathcal{L}_{\ell}^{\Sigma}$ and $\mathcal{L}_{\mathbf{P}}^{\Sigma}$.
In Section \ref{sec:forbidden_words} we will extend the ranking function for necklaces to the language $\mathcal{L}_{\ell}^{\Sigma \setminus F}$.

% In this section we give two algorithms for
% %Here two algorithms will be presented for
% the $k$-Centre Problem on the languages $\mathcal{L}_\ell^{\Sigma}$, $\mathcal{L}_\ell^{\Sigma \setminus F}$, and $\mathcal{L}_{\mathbf{p}}^{\Sigma}$.
% Both algorithms have a building block the connections between necklaces and de Bruijn sequences.
% %will be used as a building block.
% The first algorithm works as a generic tool for all three languages, while the latter will be specific to $\mathcal{L}_\ell^{\Sigma}$, with improved bounds relative to $\lambda$.
%\argy{shall we say that the first algorithm provides a generic framework that can be further utilised for more constrained languages?}

\subsection{The general algorithm}
\label{sec:prefix_tree}

The underlying idea behind the first algorithm is of building samples based on covering different possible prefixes of necklaces.
The reasoning behind this is twofold: first the set of prefixes for necklaces is much more limited than it is for unconstrained words, and secondly by covering all prefixes up to a given length $\lambda$, all words are guaranteed to share a subword of length $\lambda$ with some sample.

In order to do this effectively, it will be important to compute how many necklaces have a given prefix.
This may be done by \emph{ranking} the necklace.
The rank of a necklace  is the number of necklaces for which the canonical representation is lexicographically smaller than it.
The first algorithm to rank necklaces was given by Kopparty et. al. \cite{Kopparty2016} without a tight bound, followed by an algorithm by Sawada and Williams \cite{Sawada2017} who provided an $O(\ell^3)$ time algorithm.
Sawada and Williams show how this may be used to find the number of necklaces with a given prefix, by computing the difference between the ranks of smallest and largest necklace with the given prefix, both of which may be done in quadratic time.
This ranking function has been further extended to the fixed content case by Hartman and Sawada \cite{Hartman2019}.\\

\noindent
{\bf The k-centres selection based on a tree of  necklace prefixes:}
%At a high level, this algorithm words as follows.
%
The algorithm recursively builds the tree of possible necklace prefixes, starting with the empty string, in a breadth first manner, continuing until there are $k$ such prefixes.
Once these prefixes have been generated, the centres can be built as necklaces containing these prefixes.

This is achieved as follows.
At each step there is the set of prefixes $P$ of a length $l$ such that the number of prefixes is less than $k$.
Observe that every prefix in the set of prefixes of length $l + 1$, $P'$, will consist of a prefix from $P$ followed by a character.
Note also that every $p \in P$ must be the prefix for at least one member of $P'$.
Therefore to generate $P'$, each prefix in $p$ must be considered.
Given $p \in P$ and $\sigma \in \Sigma$, $p \sigma$ will be in $P'$ if and only if it is the prefix of a necklace.
To determine this property, the \emph{rank} of the smallest and largest non-cyclic words starting with $p \sigma$ amongst the set of necklaces can be used.
The rank of a word $w$ amongst the set of necklaces will be denoted $rank(w)$
Let $p\sigma 1^{\ell - l - 1}$ denote the smallest word starting with $p\sigma$, i.e. the word consisting of $p\sigma$ followed by $\ell - l - 1$ copies of the smallest character, and $p\sigma q^{\ell - l - 1}$ denote the largest word starting with $p\sigma$.
The number of necklaces sharing the prefix $p\sigma$ will be given by $rank(p\sigma q^{\ell - l - 1}) - rank(p\sigma 1^{\ell - l - 1})$.
If there are no such necklaces, then $p\sigma$ will be discarded, otherwise it will be added to $P'$.
The set $P'$ will be generated by repeating this process for every $p \in P$, $\sigma \in \Sigma$.
Once the size of $P'$ is greater than $k$, the algorithm will terminate using the prefixes in $P$ as a basis.
For each $p \in P$, a centre will be generated by appending an arbitrary subword following the prefix.

\begin{lemma}
There exists a polynomial-time algorithm to construct k centres of $\mathcal{L}_{\ell}^{\Sigma}$ such that every word in $\mathcal{L}_{\ell}^{\Sigma}$ shares a common substring of length at least 
%\left(
$\log_q k - 1$ 
%\log_q(k(q - 1))
%\right) 
 with the nearest centre, and therefore will be at a distance of no more than $\frac{2\ell^2}{\log^2_qk}$ from the nearest centre.
\label{thm:alg_1}
\end{lemma}

\begin{proof}

Let $\lambda$ be the length of the prefixes at the termination of the algorithm.
To bound the length of $\lambda$, observe that each sample corresponds to a prefix of length $\lambda$.
Therefore, this becomes the problem of determining the largest value of $\lambda$ such that the size of the set is less than $k$.
An upper bound on the size of the set of necklace prefixes of length $\lambda$ can be taken as the sum of the upper bound on the number of Lyndon words of length 1 to $\lambda$.
This gives the $\sum_{i = 1}^{\lambda} \frac{q^i}{i}$.
Ignoring the divisor by $i$ allows the upper bound to be rewritten into the inequality $\frac{q(q ^\lambda - 1)}{(q - 1)} \leq k$.
Note that $\log_q(k(q - 1))-1=\log_q k+ \log_q(q - 1))-1$, but $\log_q(q - 1)<1$ for any $q\geq 2$, so if we are considering only integer values then  $\floor*{\log_q(k(q - 1))-1}=\log_q(k) - 1$.
Lemma \ref{lem:max_diistance} gives a lower bound on the distance between every word in $\mathcal{L}_\ell^{\Sigma}$ and the nearest centre in the sample of $\frac{2\ell^2}{(\log_qk-1)(\log_qk-2)}$,  which is bounded by $\frac{2\ell^2}{\log^2_qk}$.

To show that the method will terminate in polynomial time, we note that the time to compute the number of necklaces with a given prefix will be $O(\ell^2)$.
For every $i \leq \lambda + 1$, at most $k$ samples will be checked.
To determine the longest $\lambda$, let there be $p_i$ prefixes of length $i$.
Observe that there will be at least $p_i + q - 1$ words of length $i = 1$.
Therefore the number of prefixes of length $i$ will be at least $(i + 1)q - i$.
Therefore the longest length will be $\frac{k - q}{q - 1}$.
Thus the maximum number prefixes that need to be checked will be $k\cdot \frac{k - q}{q - 1}$ and the total complexity will be $O\left(k\cdot \frac{k - q}{q - 1}\ell^2\right)$.
\end{proof}

%\begin{corollary}
%The algorithm described in Theorem \ref{thm:alg_1} will ensure that no word in $\mathcal{L}$ will be at a distance of more than $\frac{2\ell^2}{(\log_q(k(q - 1)) - 1)(\log_q(k(q - 1)))}$ from the nearest centre.
%\end{corollary}
\begin{theorem}\label{prefix_tree}
Problem 1 
can be solved by the polynomial time approximation algorithm for a language $\mathcal{L}_{\ell}^{\Sigma}$ 
with an approximation factor $O (log^2_{k} \ell)$  
and a language $\mathcal{L}_{\leq \ell}^{\Sigma}$ with
$O(\ell \cdot \log^2_k(\ell))$.
%respectively. 
%and for the language $\mathcal{L}_{\leq \ell}^{\Sigma}$ with an approximation factor of $O\left(log^2_{k} (l^2)\right)$ 
%\igor{If we recalculate properly based on the correction in Lemma 3 I think the ratio will be $O ((2\cdot log_kl+1)^2)$=$O(4\cdot log^2_kl)$ }.
\end{theorem}
\begin{proof}
%\igor{By the way we can also to change the base of a log here:
%$O\left(\frac{\log^2_q(l^2k)}{\log^2_q(kq)}\right)$ =$O (log^2_{k\cdot q} (l^2k))$}
Recall, from Lemma \ref{lem:min_diistance}  the lower bound on the value $\mathfrak{D}_{\mathcal{L}_{\ell}^{\Sigma}, k}$ is $\frac{\ell^2}{\log_q(\ell^2k)(\log_q(\ell^2k) + 1)}$.
Dividing the bound from Lemma \ref{thm:alg_1} by the lower bound on the distance from Lemma \ref{lem:min_diistance} gives a performance ratio $\frac{2\log_q(\ell^2k)(\log_q(\ell^2k) + 1)}{\log^2_qk}$.
Then in the big O notation it can be simplified to $O\left(\frac{\log^2_q(\ell^2k)}{\log^2_q(k)}\right)$
=$O\left(\frac{\log^2_q\ell}{\log^2_qk}\right)$
and then to $O\left(log^2_{k} \ell \right)$ by changing the base to $k$.
Therefore Problem \ref{prob:k_sample} can be solved in polynomial time within a factor of $O\left(log^2_{k} \ell \right)$
for $\mathcal{L}_\ell^{\Sigma}$ and by Lemma \ref{lem:max_length_to_fixed_length} $O(\ell(\log^2_k(\ell)))$ for $\mathcal{L}_{\leq\ell}^{\Sigma}$.
\end{proof}

\begin{lemma}
There exists a polynomial time algorithm for the $k$-centre problem on $\mathcal{L}_{\mathbf{P}}^{\Sigma}$ that can ensure that no word in $\mathcal{L}_{\mathbf{P}}^{\Sigma}$ shares a common substring of length at least $\log_q(k(q - 1)) - 1$ with the nearest centre.
Furthermore, every word in $\mathcal{L}_{\mathbf{P}}^{\Sigma}$ will be no more than  $\frac{2\ell^2}{\left(\log_q^2(k(q - 1))\right)}$ from the nearest centre.
\label{lem:fixed_content_lower_bound}
\end{lemma}

\begin{proof}
This follows from the same methods given for $\mathcal{L}_\ell^{\Sigma}$.
Using the ranking function for fixed content necklaces given by a straight forward adaptation of the ranking function given by Hartman and Sawada \cite{Hartman2019}, the set of prefixes may be generated in the same way as before.
The primary difference in these settings is that the number of necklaces and prefixes thereof are considerably smaller.
Once a set of prefixes is generated, the centres can be generated in the same 
way
%manner 
as before, with the added constraint that the word satisfies the Parikh vector.

Let $\lambda$ be the length of the longest substring shared by the word in $\mathcal{L}_{\ell}^{\Sigma}$ that is furthest from one of the centres.
By the same arguments as in Theorem \ref{thm:alg_1}, a lower bound of $\lambda$ can be given as $\lambda \geq \left(\log_q(k(q - 1))\right) - 1$.
However, unlike in the general case, all necklaces start with the same first character, increasing the lower bound on $\lambda$ to $\lambda \geq \left(\log_q(k(q - 1))\right)$.
Furthermore, as every word in $\mathcal{L}_{\mathbf{P}}^{\Sigma}$ shares the same number of occurrences of every character, therefore the size of the intersection in the overlap coefficient will be $\ell + \frac{\left(\log_q(k(q - 1))\right)\left(\log_q(k(q - 1)) + 1\right)}{2} - \left(\log_q(k(q - 1))\right)$.
Therefore the lower bound on the distance will be no more than $\frac{2\ell^2}{2\ell + \left(\log_q^2(k(q - 1))\right) - 2\left(\log_q(k(q - 1))\right)}$
and then
bounded by $\frac{2\ell^2}{\left(\log_q^2(k(q - 1))\right)}$.
%.
%This may be further simplified to $\frac{2\ell^2}{\left(\log_q^2(k(q - 1))\right)}$.
\end{proof}

\begin{theorem}
Problem \ref{prob:k_sample} can be solved for a language $\mathcal{L}_{\mathbf{P}}^{\Sigma}$ can be solved by an approximation factor of $O\left(\frac{\sqrt[q/2]{k \cdot \ell}}{\log_q(k\cdot q)^2}\right)$.
\end{theorem}
\begin{proof}
First we must establish a lower bound on the distance for this setting.
Note than, unlike for $\mathcal{L}_{\ell}^{\Sigma}$, every word contains the same set of characters.
Therefore size of the intersection in the overlap co-efficient will be at least $\ell$ for every word.
%Further than that, 
Moreover,
if any character occurs more than $\frac{\ell}{2}$ times, there will be a subword of length 2 shared by every word.
% For a given Parikh vector $\mathbf{P}$, there are at least $\frac{{\ell \choose P_1, P_2, \hdots, P_q}}{\ell}$ fixed content necklaces.

Let $\lambda$ be the length of the longest subword such that every word in $\mathcal{L}_{\mathbf{P}}^{\Sigma}$ share a subword of length at least $\lambda$ with at least one of the centres.
An upper bound on the length of $\lambda$ comes from the number of prefixes to fixed content necklaces.
For any $q \leq l \leq \ell$, there will be at least $\frac{l!}{\max(1, l - q + 1)!}$ possible subwords of length $l$.
As there are $k \cdot \ell$ possible subwords of length $\lambda$, this requires $k \cdot \ell \geq \frac{\lambda!}{\max(1,\lambda - q + 1)!}$.
assuming $\lambda \geq q$, $\frac{\lambda!}{\max(1,\lambda - q + 1)!}$ may be approximated as $(\lambda - q + 1)^q$, giving $k \cdot \ell \geq (\lambda - q + 1)^q$ bounding $\lambda$ as $\lambda \leq \sqrt[q]{k\cdot \ell} + q - 1$.

For any combination of subwords of length less than or equal to $\lambda$ guaranteeing that every word in $\mathcal{L}_{\mathbf{P}}^{\Sigma}$ has an intersection of size $ \lambda(\lambda + 1)$ requires $\lambda^2 \leq \sqrt[q]{k\cdot\ell} + q - 1$, however for any $\lambda \geq 2$, this contradicts the assumption that $\lambda$ is the largest value such that $\lambda \leq \sqrt[q]{k\cdot\ell} + q - 1$.
Therefore the size of the intersection must be less than or equal to $\ell - \lambda + (\lambda)(\lambda + 1)$, which by substituting the upper bound of $\sqrt[q]{k \cdot \ell} + q - 1$ gives $\ell - \sqrt[q]{k \cdot \ell} + q - 1 + (\sqrt[q]{k \cdot \ell} + q - 1)(\sqrt[q]{k \cdot \ell} + q)$.
This gives an lower bound on the distance as $\frac{\ell^2}{\ell - \sqrt[q]{k \cdot \ell} + q - 1 + (\sqrt[q]{k \cdot \ell} + q - 1)(\sqrt[q]{k \cdot \ell} + q)}$.

To get the approximation ratio, the lower bound in the distance given in Lemma \ref{lem:fixed_content_lower_bound} by this upper bound gives $\frac{2\ell - 2\sqrt[q]{k \cdot \ell} + 2q - 2 + 2(\sqrt[q]{k \cdot \ell} + q - 1)(\sqrt[q]{k \cdot \ell} + q) }{\log^2_q(k(q - 1))}$.
Assuming that $\ell \leq \left(\sqrt[q]{k \cdot \ell + q - 1}\right)^2$, this can be simplified for big O notation to $O\left(\frac{\sqrt[q/2]{k \cdot \ell}}{\log_q(k\cdot q)^2}\right)$.
\end{proof}

\subsection{Sampling with forbidden subwords}
%with the ranking and unranking function and forbidden subwords}
\label{sec:forbidden_words}

In order to generalise the algorithm described in Theorem \ref{thm:alg_1}, the ranking and unranking functions must be generalised to account for forbidden words.
This is a much more challenging problem compared to the general case primarily due to the cyclic nature of the words.
Unlike with an non-cyclic word, when counting the number of cyclic words without a given forbidden word, it must be ensured that it does not occur for {\em any} shift, as opposed to just one.
This is further complicated when considering multiple forbidden words, where it must be checked that no forbidden word occurs for any rotation.
%\noindent
Ruskey and Sawada \cite{Ruskey2000} computed the size of $\mathcal{L}_\ell^{\Sigma}$ as
%
%\begin{align*}
$
    N_q^\ell(F) = \sum\limits_{d | \ell} \phi(d) C_q^{\frac{\ell}{d}}(F),
    \label{eq:1d_necklaces_forbidden}
$
%\end{align*}
%
where $C_q^{\ell}(F)$ is a function for counting the number of cyclic words of length $\ell$ containing no subword in $F$.
This can be computed in polynomial time for a constant size of $|F|$.
The number of Lyndon words of length $\ell$ with not containing any subword in $F$, denoted $L_q^\ell(F)$, is given relative to the number of Necklace, using the function:
\begin{align}
    L_q^\ell(F) = \sum\limits_{d | \ell} \cdot \mu(d)N_q^{\ell/d}(F).
\end{align}
For the remainder of this section, let $\mathbf{N}_q^\ell(F)$ and $\mathbf{L}_q^{\ell}(F)$ denote the sets of necklaces and Lyndon words respectively.

%\subsubsection{Ranking}
%
Before introducing the function for  ranking and unranking of forbidden words, some theoretical results must be established.
Let $\mathbf{T}(w,F)$ be the set of words such that the canonical representation for each every word $v \in \mathbf{T}(w,F)$ is smaller than $w$, and no forbidden word in $F$ occurs as a subword.
Let $\langle x \rangle$ denote the canonical rotation of some word $x$, and let $f \nsubseteq w$ denote that $f$ is not a subword of $w$.
Using this notation we get $\mathbf{T}(w,F) = \{x \in \Sigma^{\ell} : \langle x \rangle < w, \forall f \in F, f \nsubseteq x\}$

Three further sets are needed for the purpose of ranking.
The first of these is the set of aperiodic words such that the smallest rotation is less than some given word $w$, denoted $\mathbf{T}'(w,F) = \{ x \in \Sigma^\ell : \langle x \rangle < w, \forall f \in F, f \nsubseteq x, x \text{ is aperiodic}\}$.
Next is the set of words on length $l \leq n$ where the smallest rotation is less than $w$, denoted $\mathbf{T}_l(w,F) = \{x \in \Sigma^{l} : \langle x \rangle < w, \forall f \in F, f \nsubseteq x\}$.
For two words $x$ and $w$ of lengths $l$ and $n$ respectively, $\langle x \rangle < w$ if and only if $\langle x^n \rangle < w^l$.
The final set is that of aperiodic words of a given length $l$ where the canonical representation is less than, denoted $\mathbf{T}'_l = \{  x \in \Sigma^l : \langle x \rangle < w, \forall f \in F, f \nsubseteq x, x \text{ is a Lyndon word}\}$.

\begin{lemma}
For every $d$ that is a factor of $l$, the size of $\mathbf{T}_l(w,F)$ is equal to $\sum\limits_{d | \ell} |\mathbf{T}_d'(w,F)|$.
\label{lem:size_of_T_F}
\end{lemma}

\begin{proof}
Observe that every word in the set $\mathbf{T}_l(w,F)$ will either be aperiodic, in which case it will belong also to the set $\mathbf{T}_l'(w,F)$, or it will be periodic.
If it is periodic, the period must be some value that is a factor of $l$.
Given some word with a period $d$, if it is smaller than $w$, then it will occur in the set $\mathbf{T}_d'(w,F)$.
By definition, any word greater than $w$ will not occur in any set $\mathbf{T}_b'(w,F)$ for any $b$ such that $\ell \bmod{b} \equiv 0$.
As each set $\mathbf{T}_d'(w,F)$ consists only of aperiodic words, there can be no word that occurs in both $\mathbf{T}_d'(w,F)$ and $\mathbf{T}_e'(w,F)$ for $d \neq e$.
Therefore the size of $\mathbf{T}_l(w,F)$ can be computed as $|\mathbf{T}_l(w,F)| = \sum\limits_{d | l} |\mathbf{T}_d'(w,F)|$.
\end{proof}
By application of the M\"{o}bius inversion formula to $|\mathbf{T}_l(w,F)| = \sum\limits_{d | l} |\mathbf{T}_d'(w,F)|$, an equation for the size of $\mathbf{T}'_{l}(w,F)$ can be derived as:

\begin{align}
    |\mathbf{T}'_l(w,F)| = \sum\limits_{d | l} \mu\left(\frac{\ell}{d}\right)|\mathbf{T}_d(w,F)|. \label{eq:compute_T_p_F}
\end{align}

This can be used to rank some word $w$ amongst the set of Lyndon words without any forbidden subword.
%Observe the following:

\begin{lemma}
The number of Lyndon word smaller than some word $w$ without any forbidden subword in $F$ will be given by $rank_L(w,F) = \frac{1}{\ell} \cdot \mathbf{T}'(w,F)$.
\label{lem:ranking_L_F}
\end{lemma}

\begin{proof}
As every Lyndon word is aperiodic, it has $n$ unique rotations.
Therefore, for any given word $w$, each Lyndon word will occur $n$ times within the set of aperiodic words with some rotation smaller than $w$, if and only if the canonical representation of the Lyndon word is smaller than $w$.
\end{proof}
%
%The number of necklaces smaller than $w$ may be computed using $rank_L(w,F)$ with the following Lemma:
%
In the next lemma we compute the number of necklaces smaller than $w$ using $rank_L(w,F)$.

\begin{lemma}
The number of necklaces smaller than $w$ without any forbidden subword in $F$ is equal to $rank_N(w,F) = \sum\limits_{d | \ell} \frac{1}{d} \cdot \mathbf{T}'_d(w, F)$.
\label{lem:ranking_N_F}
\end{lemma}

\begin{proof}
It follows from Lemma \ref{lem:size_of_T_F} that all necklaces smaller than some word will either be aperiodic, or periodic with a period that is a factor of the length of the necklace.
From Lemma \ref{lem:ranking_L_F}, the necklaces that are smaller than $w$ and are aperiodic are $\frac{1}{\ell}\cdot \mathbf{T}'_d(w, F)$.
Similarly, the necklaces with a period of some factor $d$ of $n$ are $\frac{1}{\ell}\cdot \mathbf{T}'_d(w, F)$.
\end{proof}

The problem now becomes computing $|\mathbf{T}_l(w,F)|$.
To do this, the set will be partitioned to the set $\mathbf{A}_w(t,j,F)$ such that for every word $v \in \mathbf{A}_w(t,j,F)$ the following hold.

\begin{itemize}
    \item $t$ is the smallest cyclic shift such that shifting $v$ by $t$, denoted $v \cdot t $, makes the resulting word smaller than $w$, i.e. $v \cdot t < w$.
    \item Under the shift by $t$, $j$ is the length of the longest prefix of $v \cdot t$ that is also a prefix of $j$.
\end{itemize}

\begin{lemma}
The size of $\mathbf{A}_w(t,j,F)$ can be computed in $O(q\ell^{|F| + 2})$ time.
\label{lem:ranking_A}
\end{lemma}

\begin{proof}
This can be done by considering two possible cases for the set.
First is the case $t + j \leq \ell$.
In this case, every word will be of the form $\beta w_1 \hdots w_j x \rho$ where:

\begin{itemize}
    \item $\beta$ is some word of length $t$ with no forbidden subword such that every suffix is greater than $w$;
    \item $w_1 \hdots w_j$ is the prefix of $w$ of length $j$;
    \item $x$ is a character smaller than $w_{j + 1}$;
    \item $\rho$ is a word with no restrictions other than having no forbidden substrings.
\end{itemize}

To compute the number of possible words satisfying $\rho$ and $\beta$, we define the function $B'_{\alpha}(l, j, t, P, S)$.
A full definition of this function is given in Appendix \ref{app:B_function}.
At a high level, this function works by recursively checking how many possible ways of extending the string based on the sets $P$ and $S$. 
$S$ represents the set of suffixes of $\alpha$ that are prefixes of $w_1 \hdots w_j$.
$P$ initially represents the set of prefixes of $\alpha$ that are suffixes of $\alpha_1 \hdots \alpha_j$.
As this function can be computed recursively, the computational complexity will be equal to the number of potential calls to $B'$.
There are $\ell$ possible value for $l$ and $t + j$, and $\ell^{|F|}$ for both $S$ and $P$.
This gives a total time complexity of This will take $O(q\ell^{|F| + 3})$ time to compute in the worst case.
Alongside this, two auxiliary functions $\theta(w_1 \hdots w_j \sigma)$ and $\Omega(w_1 \hdots w_j \sigma)$ are needed. These, respectively, compute the sets of prefixes and suffixes of forbidden words of $w_1 \hdots w_j \sigma$, for some character $\sigma$ .
Using the above, we get:
$$
|\mathbf{A}_w(t,j,F)| = \sum\limits_{\sigma = 1}^{w_{j + 1} - 1} B'_{w}(\ell - j - 1, 0, \ell - (j + t + 1), \theta(w_1 \hdots w_j \sigma), \Omega(w_1 \hdots w_j \sigma))
$$
In the second case, every word will be of the form $w_s w_{s + 1}\hdots w_{j} x \beta w_1 \hdots w_{s - 1}$.
Let $\delta$ be the length of the longest prefix of $w$ that is a suffix $w_s \hdots w_j$.
If $x < w_{\delta + 1}$, then the shift by $s - \delta$ would be smaller than $\alpha$.
Therefore $x$ must be greater than or equal to $w_{\delta + 1}$.
From this, the size of $\mathbf{A}_w(t,j,F)$ can be computed as:
%
%\begin{align*}
$|\mathbf{A}_w(t,j,F)| = B'_{w}(\ell - j - 1, \delta + 1, 0, \theta(w_1 \hdots w_{\delta + 1}), \Omega(w_1 \hdots w_{\delta + 1})) +
\sum\limits_{\sigma = w_{\delta + 1} + 1}^{w_{j + 1} - 1} B'_{w}(\ell - j - 1, 0, 0, \theta(w_1 \hdots w_{\delta + 1}), \Omega(w_1 \hdots w_{\delta + 1}))$

%\end{align*}
%
As $B'_{\alpha}(l, j, t, P, S)$ can be computed in $O(q\ell^{|F| + 2})$ and stored for all value of values of $l,j,t,P$ and $S$, the time to compute the size of $\mathbf{A}_w(t,j,F)$ will be $O(q)$.
As this is dominated by the time to compute $B'$, the total time will be $O(q\ell^{|F| + 3})$.
\end{proof}

\begin{theorem}
The rank of a word amongst all necklaces without any forbidden subword may be computed in $O(q\ell^{|F| + 2}\log_2^2(\ell))$.
\label{thm:ranking_F_complexity}
\end{theorem}

\begin{proof}
It follows from Lemma \ref{lem:ranking_A} that by separately computing the values of $B'$, the time to compute the size of $\mathbf{A}_w(j,t,F)$ will be $O(q\ell^{|F| + 2})$.
Following Lemma \ref{lem:ranking_N_F}, the number of Necklaces may be computed by summing the size of $\mathbf{T}'_d(w,F)$ for every factor $d$ of $\ell$.
Note that there are at most $\log_2(\ell)$ factors of $\ell$.
The size of $\mathbf{T}'_d(w,F)$ can be computed using Lemma Equation \ref{eq:compute_T_p_F}.
In the worst case there will be $O(\log_2\ell)$ sets of $\mathbf{T}_d(w, F)$, each of which taking at most $O(q\ell^{|F| + 2})$ time to compute.
Putting this together, the total time complexity will be $O(q\ell^{|F| + 2}\log_2^2(\ell))$
\end{proof}

\begin{theorem}
For a constant number of forbidden words, Problem \ref{prob:k_sample} for $\mathcal{L}_{\ell}^{\Sigma \setminus F}$ can be solved in polynomial time with an approximation factor of $O(\log_k^2(\ell))$.
\end{theorem}

\begin{proof}
The same approach given in Section \ref{sec:prefix_tree} may be used to build centres.
In this case the ranking function described in this section may be used in place of the ones used there.
Clearly the same bounds on length of the prefix will hold in the worst case, giving an upper bound on the distance of $\frac{2\ell^2}{\log_q^2(k)}$.
A lower bound follows from the same observation that the length of the longest common subword between the furthest word in the language and the nearest centre will be bound from above by the number of Lyndon words without any forbidden subwords.
A bound on this is $\frac{q^{\lambda} - (\ell|F|q^{\lambda - 2})}{\ell}$, which will be of order $O(\frac{q^\lambda}{\ell})$ for a constant size $|F|$. 
% For a fixed size of $|F|$, observe that the upper bound on the distance may be calculated in the same way as the language will still grow exponentially.
% Similarly the upper bound on $\lambda$ will be bounded from bellow by the size of the set of necklaces, which will be at least $\frac{q^{\lambda}}{\ell} - $
Thus the approximation factor between the bound given by this algorithm will be the same as in the unconstrained case, giving a factor of $O(\log_k^2(\ell))$.
\end{proof}

\section{Sampling via de Bruijn sequences}
%\argy{we do not start smoothly. can we rewrite the first paragraph and sell more our result? This should be done after Igor's comments are fixed, since they are more crucial}
%\igor{OK}
The primary issue with the prenecklace based algorithm is that it does not take advantage of any additional space left in the samples.
As such a different approach will have to be considered to build the samples.
Following the same motivation of maximising the length of the subword shared between every word in the language and the nearest centre, observe that this requires every subword must occur at some point in the sample.

For a given length $\lambda$, there are $q^{\lambda}$ subwords.
A de Bruijn sequence of order $\lambda$ is a word of length $q^{\lambda}$ where every word of length $\lambda$ occurs as a subword.
This makes it a natural candidate for use as the basis for the centres.

\begin{figure}
    \centering
    0000001000011000101000111001001011001101001111010101110110111111
    
   % \centering
    \begin{tabular}{l|l}
        Centre & Word \\
        1 & {\color{red} 000000}1000011000{\color{blue} 10100}\\
        2 & \hspace{2.6cm}{\color{blue} 10100}01110010010{\color{darkgreen} 11001}\\
        3 & \hspace{5.2cm}{\color{darkgreen} 11001}10100111101{\color{purple}01011}\\
        4 & \hspace{7.8cm}{\color{purple}01011}0110111111{\color{red}000000}
    \end{tabular}
    \caption{Example of how to split the de Bruijn sequence of order 6 between 4 samples.
    Highlighted parts are the shared subwords between two samples.}
    \label{fig:deBrujinExample}
\end{figure}

\begin{lemma}
There exists an algorithm with a worst case running time of $O(\ell k)$ for the $k$-centre problem on $\mathcal{L}_{\ell}^{\Sigma}$ ensuring that every word in $\mathcal{L}_{\ell}^{\Sigma}$ shares a common substring of length at least $\log_q(k)$.
Further this will ensure that no word in $\mathcal{L}_{\ell}^{\Sigma}$ is a distance of more than $\frac{2\ell^2}{\log_q^2(k \ell)}$ from the nearest centre.
\label{thm:alg_2}
\end{lemma}

\begin{proof}
The main idea of this algorithm is to take the de Bruijn sequence of order $\lambda$ and divide it between samples while ensuring that all subwords of length $\lambda$ occur at some point as a subword of a sample.
Note that the length of the de Bruijn sequence of order $\lambda$ will be $q^{\lambda}$.
The de Bruijn sequence may be efficiently generated in time linear to the length of the sequence \cite{Ruskey2000}, which must be less than $\ell \times k$.

Naively splitting the sequence between the $k$ centres may lead to subwords being lost.
In order to account for this, the sequence may be split into samples of size $\ell - \lambda + 1$.
The first centre can be generated by taking the first $\ell$ characters of the de Bruijn sequence.
To ensure that every subword of length $\lambda$ occurs, the fist $\lambda - 1$ characters of the second centre will be the same as the the last $\lambda - 1$ character of the first centre.
Repeating this, the $i^{th}$ centre will be the subword of length $\ell$ starting at position $i(\ell -\lambda)$ in the de Bruijn sequence.
An example of this is given in Figure 
 \ref{fig:deBrujinExample}.

To determine the length of $\lambda$ relative to $k$ and $\ell$, note that the size of the corresponding de Bruijn sequence must be small enough such that every subword may occur.
Formally, for $\lambda$ to be feasible, $q^{\lambda} \leq k(\ell - \lambda + 1)$.
This may be rearranged in terms of $\lambda$, giving as an upper bound $ \leq \log_q(k\ell)$.
%\igor{<$\log_q (k\ell)$}.

Using these, the centres may be made be formed by taking each of these samples, and appending the first $\lambda - 1$ of the next sample.
For a given $k$ centres of length at most $\ell$, $\lambda$ will be the largest value such that $q^{\lambda} \leq k(\ell - \lambda + 1)$, which may be rewritten as $\lambda \leq \log_q(k (\ell - \lambda + 1))$<$\lambda \leq \log_q(k\ell)$.
%\igor{<$\log_q (k\ell)$}.
%
An upper bound on the distance may be gained using Lemma \ref{lem:max_diistance}, giving $\frac{2\ell^2}{(\log_q(k\ell))(\log_q(k\ell) + 1)}$, which may be 
%simplified to
bounded by
$\frac{2\ell^2}{\ceil{ \log_q^2(k\ell})}$.

In terms of complexity, there are known algorithms to output the de Bruijn sequence in $O(q^{\lambda})$ time, which by the definition of $\lambda$ will be $O(k\ell)$ time in the worst case.
From the sequence, the process of dividing into $k$ samples will take no more than time linear to the size of the sequence, giving a total complexity of $O(k\ell)$.
It is worth noting that any algorithm that outputs the centres will have a complexity of at least $O(k \ell)$.
\end{proof}

%Clearly for $\lambda \leq \frac{\ell}{2}$ \igor{ WHY?} this will preform better than the algorithm presented in Section \ref{sec:prefix_tree}.

\begin{theorem}
The algorithm described in Lemma \ref{thm:alg_2} will approximate the optimal solution within a factor of $8$.
\end{theorem}

\begin{proof}
Recall from Lemma \ref{lem:max_diistance} that the minimum distance the word that is furthest from the sample can be is $\frac{\ell^2}{\log^2_q(k\ell^2)}$.
The upper bound on distance given by Lemma \ref{thm:alg_2} is $\frac{2\ell^2}{\log_q^2(k\ell)}$.
Dividing this upper bound by the lower bound gives $\frac{2\log^2_q(k\ell^2)}{\log_q^2(k\ell)}$.
Rewriting this in terms of base $k$ gives $\frac{2\log_k^2(k\ell^2)}{\log_k^2(k\ell)}$
= $2(\frac{\log_k(k\ell^2)}{\log_k(k\ell)})^2$
= $2(\frac{1+\log_k(\ell^2)}{1+\log_k(\ell)})^2$
= $2(\frac{1+2\log_k(\ell)}{1+\log_k(\ell)})^2$
= $2(\frac{1+2\log_k(\ell)+1-1}{1+\log_k(\ell)})^2$=\\
= $2(2-\frac{1}{1+\log_k(\ell)})^2$.
If $(2 - \frac{1}{1 + \log_k(\ell)})^2 \leq 4$ then this algorithm will approximate the optimal solution within a factor of $8$. % \argy{I think it is always the case since $\frac{1}{1 + \log_k(\ell)} < 1$}
As $\frac{1}{1 + \log_k(\ell)}$ will be no more than 1, this will hold.
% There are two boundary conditions on this, the first of which is that $2 - \frac{1}{1 + \log_k(\ell)} \geq -2$, which can be rewritten as $\frac{1}{1 + \log_k(\ell)} \leq 4$.
% This means that for any $k$ and $\ell$ satisfying the inequality $1 \leq 4 + 4\log_k(\ell)$, the function will be satisfied.
% This can be rearranged to give $\log_k(\ell) \geq -\frac{3}{4}$.
% Therefore for $\ell \geq k^{\frac{-3}{4}}$ this will give an approximation factor of no more than $8$.
% The second condition is that $2 - \frac{1}{1 + \log_k(\ell)} \leq 2$, meaning that $\frac{1}{1 + \log_k(\ell)} \geq 0$.
% This implies that $1 + \log_k(\ell) \geq 0$, meaning that $\ell \geq \frac{1}{k}$.
% Therefore this algorithm admits an approximation factor of 8 if and only if $\ell \geq k^{-\frac{3}{4}}$.
% %
% If $\ell$ is greater than or equal to $k^{-\frac{3}{4}}$ then the equation $2(2 - \frac{1}{1 + \log_k(\ell)})$ will have an approximation ratio bounded by $(\frac{1}{1 + \log_k(\ell)})^2$, which is proportional to $\log^2_{\ell}(k)$.
%
%For $k \geq \ell$, this will give an approximation factor of $8$.
%
%This can be rewritten as $\frac{2(1 + \log_k(\ell^2))^2}{(1 + \log_k(\ell))^2}$
%and may be further rewritten to give \\
%$\frac{2\log^2_k(\ell^2) + 4\log_k(\ell^2) + 2}{\log_k^2(\ell) + 2\log_k(\ell) + 1}$=$ %\frac{8\log^2_k(\ell) + 8 \log_k (\ell) + 2}{\log^2_k(\ell) + 2\log_k(\ell) + 1}$.
%This gives an approximation factor of $8$.
\end{proof}

\bibliography{bib.bib}
\bibliographystyle{plainurl}

\newpage
\appendix
\section{Ranking functions}
\label{app:B_function}

The function $B'_{\alpha}(l,j,t,P,S)$ can be thought of as counting the size of the set of words where, for every $v$ in the set:
\begin{itemize}
    \item $v$ has length $t$.
    \item For every $p \in P$, and $s\in S$, no forbidden subword occurs in the word $p v s$.
    \item The first $j$ characters of $v = w_1 w_2 \hdots w_j$.
    \item Every suffix of the subword $v_{t - l} v_{t - l  + 1} \hdots v_t$ is greater than $w$.
\end{itemize}

To compute $B'_{\alpha}(l,j,t,P,S)$, two auxiliary functions will be introduced to compute the $P$ and $S$ after the next character $\sigma$ is introduced.
$\theta(P, \sigma)$ will take as argument the set of prefixes $P$ and some character $\sigma$, and returning the set of prefixes of $F$ where either $\sigma$ is the first character of a forbidden word or there exists some $p \in P$ where $p \sigma \subseteq F$ - i.e. the prefixes that may be continued by adding this character.
$\Omega(S, \sigma)$ will return the members of $S$ that are suffixes of some forbidden word, and the words $s \sigma$ for $s \in S$ where $s \sigma$ is a subword of a forbidden word.
To compute $B'_{\alpha}(l,j,t,P,S)$, the set may be further partitioned by the next character $\sigma$.
If $l > 0$, then there is no lower bound on the value of $\sigma$, otherwise $\sigma$ must be greater than or equal to ${\alpha}_{j + 1}$.
For each $\sigma$ the number of words for which the $(j + 1)^{th}$ character equals $\sigma$ will be as follows:

\begin{itemize}
    \item If there exists some $p \in P$ such that $p \sigma = f$ for some forbidden word $f$ then there will be no subsequent words.
    \item Otherwise, if $t - j = 1$ then if there exists some $p \in P$ and some $s \in S$ where $p \sigma s = f$ or $\sigma s = f$ then there will also be no words in this set, otherwise there will be only 1 word.
    \item If $l > 0$ and $t - j > 1$ then the size of the set will be equal to $B'_{{\alpha}}(l - 1, t - t,j, \theta(P, \sigma), \Omega(S,\sigma))$.
    \item If $l = 0$ and $\sigma > {\alpha}_{j + 1}$ then the size of the set will be equal to the size of the set $B'_{\alpha}(0, t - 1, 0, \theta(P, \sigma), \Omega(S, \sigma))$.
    \item Otherwise, $l = 0$ and $\sigma = {\alpha}_{j + 1}$, therefore the size of this set will be equivalent to the set $B'_{\alpha}(0, t, j + 1, \theta(P, \sigma), \Omega(S, \sigma))$.
\end{itemize}

Therefore $B'_{\alpha}(l,t,j,P,S)$ may be computed recursively as follows: $B'_{\alpha}(l,t,j,P,S)=$

\[=
\begin{cases}
    0 & t = j > 0\\
    1 & t = j = 0\\
    \sum\limits_{\sigma = 1}^{q} \begin{cases}
        B'_\alpha(l - 1, t - 1, 0, \theta(P, \sigma), \Omega(S, \sigma)) & \nexists p \in P, f \in F \text{ s.t. } p \sigma = f\\
        0 & Otherwise.
    \end{cases} & l > 0, t > 1\\
    \left(\sum\limits_{\sigma = \alpha_{j} + 1}^{q} \begin{cases}
        B'_\alpha(l - 1, t - 1, 0, \theta(P, \sigma), \Omega(S, \sigma)) & \nexists p \in P, f \in F \text{ s.t. } p \sigma = f\\
        0 & Otherwise.
    \end{cases}\right)\\
    + \left(\begin{cases}
        B'_{\alpha}(0, t, j + 1, \theta(P, \alpha_{j + 1}), \Omega(S, \alpha_{j + 1})) & \nexists p \in P, f \in F \text{ s.t. } p \sigma = f\\
        0 & Otherwise.
    \end{cases}\right) & l = 0, t - j > 1\\
    \sum\limits_{\sigma = 1}^{q} \begin{cases}
        1 & \nexists p \in P \cup \{\emptyset\}, f \in F, s \in S \cup \{\emptyset\} \text{ s.t. } p \sigma s = f\\
        0 & Otherwise.
    \end{cases} & l > 0, t = 1\\
    \left(\sum\limits_{\sigma = \alpha_{j + 1} + 1} \begin{cases}
        1 & \nexists p \in P \cup \{\emptyset\}, f \in F, s \in S \cup \{\emptyset\} \text{ s.t. } p \sigma s = f\\
        0 & Otherwise.
    \end{cases} \right)\\
    + \left(\begin{cases}
        1 & \nexists p \in P \cup \{\emptyset\}, f \in F, s \in S \cup \{\emptyset\} \text{ s.t. } p \alpha_{j + 1} s = f\\
        0 & Otherwise.
    \end{cases}\right)& Otherwise. 
\end{cases}
\]

Note that there are at most $\ell$ possible values for $l,t$ and $j$, and $\ell^{|F|}$ possible values for $P$ and $S$.
In the worst case each of these must be computed.
Assuming for some given arguments that the values of $B'_{\alpha}$ for each character has been computed, then it will take $O(q)$ time to compute the value of $B'_{\alpha}$ for these arguments.

\section{Prefix algorithm}

\begin{algorithm}
\caption{Prefix algorithm}
\begin{algorithmic}[1]
\Procedure{Prefix}{$\Sigma, n, k$}
\State Set of current prefixes, $P$ 
\State Set of new prefixes, $P'$
\State Prefix in $P$, $p$
\State Character in $\Sigma$, $\beta$
\State $P \gets \emptyset$
\State $P' \gets \emptyset$
\While{$|P| < N_k$}
    \State $P' \gets \emptyset$
    \For{$p \in P$}
        \For{$\beta \in \Sigma$}
            \If{$rank(\min(p\beta)) \neq rank(\max(p\beta))$}
                \State $P' \gets P \cup \{p\beta\}$
            \EndIf
            \If{$|P \cup P'| = k$ and $\beta \neq |\Sigma |$}
                \State \textbf{return} $P \cup P'$
            \EndIf
        \EndFor
        \State $P \gets P \setminus p$
        \If{$|P \cup P'| = k$}
            \State \textbf{return} $P \cup P'$
        \EndIf
    \EndFor
    \State $P \gets P'$
\EndWhile
\State $Prefixes \gets \emptyset$
\For{$p \in P'$}
    \State $Prefixes \gets Prefixes \cup \{unrank(\min(p) + \floor{\frac{\min(p) + \max(p)}{2}})\}$
\EndFor
\State \textbf{return} $P$
\EndProcedure
\end{algorithmic}
\label{alg:prefix_1d}
\end{algorithm}

\end{document}